\documentclass[11pt]{article}

\usepackage{mathtools}
\usepackage{microtype}
\allowdisplaybreaks

\DeclareMathOperator{\gal}{Gal}
\usepackage{seqsplit}
\usepackage{authblk}

\newcommand{\bbbq}{\mathbb{Q}}

\usepackage{amssymb, amsmath, amsthm}

\newtheorem{thm}{Theorem}
\newtheorem{lem}[thm]{Lemma}

\newtheorem{obs}{Remark}

\usepackage{fullpage}

\newcommand{\p}{\mathcal{S}}

\newcommand*\samethanks[1][\value{footnote}]{\footnotemark[#1]}

\begin{document}

\title{A Note on the Unsolvability of the\\ Weighted Region Shortest Path Problem\thanks{A preliminary version appeared at EuroCG 2012 \cite{eurocg}.}}

\author[1]{Jean-Lou De Carufel\thanks{Research supported by FQRNT.}}
\author[1,2]{Carsten Grimm}
\author[1]{Anil Maheshwari\thanks{Research supported by NSERC.}}
\author[3]{Megan Owen\thanks{Research supported by a Fields-Ontario Postdoctoral Fellowship.}}
\author[1]{Michiel Smid\samethanks[3]}

\affil[1]{School of Computer Science, Carleton University, Ottawa, Canada}
\affil[2]{Fakut\"at f\"ur Informatik, Otto-von-Guericke-Universit\"at Magdeburg, Magdeburg, Germany}
\affil[3]{Cheriton School of Computer Science, University of Waterloo, Waterloo, Canada}

\maketitle

\begin{abstract} 
Let $\p$ be a subdivision of the plane into polygonal regions, where each region has an associated positive weight. 
The {\em weighted region shortest path problem} is to determine a shortest path in $\p$ between two points 
$s,t\in \mathbb{R}^2$, where the distances are measured according to the weighted Euclidean metric---the length of a path is defined to be the weighted sum of (Euclidean) lengths of the sub-paths within each region.
We show that this problem cannot be solved in the \emph{Algebraic Computation Model over the Rational Numbers} ({\sf ACM$\bbbq$}). 
In the {\sf ACM$\bbbq$}, one can compute exactly any number that can be obtained from the rationals $\bbbq$ by applying a finite number of operations from $+$, $-$, $\times$, $\div$, $\sqrt[k]{\phantom{\cdot}}$, for any integer $k\geq 2$. 
Our proof uses Galois theory and is based on Bajaj's technique.
\end{abstract}	
 
\maketitle

\section{Introduction} \label{sec::introduction}

The weighted region shortest path problem is one of the classical path problems in Computational Geometry and 
has been studied over the last two decades. It was originally introduced
by Mitchell and Papadimitriou~\cite{mitchell1991weighted}
as a generalization of the two-dimensional shortest path problem with obstacles. There are several
well known approximation algorithms for this problem
(see~\cite{aleksandrov2005determining,bose2011survey,Mitchell98geometricshortest,mitchell1991weighted} for instance).
In this paper, we show that determining the exact shortest path distance in this setting is an unsolvable problem in an algebraic model of computation, confirming the  suspicion expressed
by Mitchell and Papadimitriou~\cite[Section~4]{mitchell1991weighted}.
Thus, we provide further justification for the search for approximate solutions as opposed to exact ones. 

The algebraic complexity of geometric optimization problems was first studied by Bajaj, who showed that Euclidean shortest paths among polyhedral obstacles in three dimensions \cite{Bajaj_tech} and solutions to the Weber problem and its variations \cite{DBLP:journals/dcg/Bajaj88} cannot be expressed as finite algebraic expressions.
More recently, the algebraic complexity of semi-definite programming \cite{SDP_alg} and shortest paths through certain cube complexes \cite{CAT0} were investigated. 
De Carufel et al.~\cite{DBLP:journals/corr/abs-1212-1617}
studied a variant of the Fr\'echet distance
that has a lower sensitivity to the presence of outliers
than the usual one.
They showed that this variant cannot be computed exactly whithin the
\emph{Algebraic Computation Model over the Rational Numbers} ({\sf ACM$\bbbq$}).
In the {\sf ACM$\bbbq$}, one can compute exactly any number that can be obtained from the rationals $\bbbq$ by applying a finite number of operations from $+$, $-$, $\times$, $\div$, $\sqrt[k]{\phantom{\cdot}}$, for any integer $k\geq 2$.
In this paper,
we employ Bajaj's technique \cite{DBLP:journals/dcg/Bajaj88} to show that
the weighted region shortest path problem is unsolvable
within the {\sf ACM$\bbbq$}. The technique is as follows.

As a consequence of the fundamental theorem of Galois \cite{dummit}, we know that there is no general formula to solve a polynomial equation of degree \(d \ge 5\) \emph{by radicals}. However, there are some polynomial equations of degree $d \geq 5$ that can be solved by radicals. The \emph{Galois group} $\gal(p)$ of an irreducible polynomial $p$ over $\bbbq$ determines the solvability of $p$ by radicals: the equation $p(x)=0$ is solvable by radicals if and only if $\gal(p)$ is \emph{solvable} (refer to~\cite{dummit}).
Intuitively,
$p$ is unsolvable by radicals if its coefficients are algebraically independent, i.\,e., not related by an algebraic expression.

We will present an instance of the weighted region shortest path problem such that solving this instance exactly within the {\sf ACM$\bbbq$} is equivalent to the statement that the polynomial equation $p_{12}(x)=0$ in Equation~\eqref{eqn:p12} is solvable by radicals.
However,
we will show that the Galois group of $p_{12}$ is $S_{12}$
(i.e., the symmetric group over $12$ elements) up to isomorphism. This is proved using the following theorem.\footnote{Alternatively, it can be verified using symbolic computation software.  For example, GAP uses the algorithm from \cite{GAP} to test the solvability of polynomials up to degree \(15\) via the command \texttt{isSolvable}, and MAGMA implements an extension of the algorithm in \cite{Geissler}, that works for polynomials of arbitrary degree, limited only by time and space constraints.}

\begin{thm}[Bajaj~\cite{DBLP:journals/dcg/Bajaj88}]
\label{theorem bajaj}
Let $p$ be a polynomial of even degree $d\geq 6$.
Suppose that there are three prime numbers $q_1$, $q_2$ and $q_3$ that do
not divide the discriminant $\Delta(p)$ of $p$, such that
\begin{eqnarray}
\label{thm bajaj cond 1}
p(x) &\equiv& p_{d}(x) \pmod{q_1} \enspace,\\
\label{thm bajaj cond 2}
p(x) &\equiv& p_1(x)p_{d-1}(x) \pmod{q_2} \enspace,\\
\label{thm bajaj cond 3}
p(x) &\equiv& p_1'(x)p_2(x)p_{d-3}(x) \pmod{q_3} \enspace,
\end{eqnarray}
where $p_d(x)$ is an irreducible polynomial of degree $d$ modulo $q_1$;
$p_{d-1}(x)$ (respectively $p_1(x)$) is an irreducible polynomial of degree
$d-1$ (respectively of degree $1$) modulo $q_2$;
$p_{d-3}(x)$ (respectively $p_1'(x)$ and $p_2(x)$) is an irreducible polynomial of degree
$d-3$ (respectively of degree $1$ and of degree $2$) modulo $q_3$. Then
$\gal(p) \cong S_d$.

If $d\geq 5$ is odd,
the same result holds if we replace (\ref{thm bajaj cond 3}) by
\begin{eqnarray*}
p(x) &\equiv& p_2(x)p_{d-2}(x) \pmod{q_4} \enspace,
\end{eqnarray*}
where $q_4$ is a prime number such that $q_4\nmid\Delta(p)$
and $p_{d-2}(x)$ (respectively $p_2(x)$) is an irreducible polynomial of degree
$d-2$ (respectively of degree $2$) modulo $q_4$.
\end{thm}
Observe that 
\eqref{thm bajaj cond 1} implies that $p(x)$ is irreducible over $\bbbq$,
which implies that $\gal(p)$ is a \emph{transitive} group.
Conditions (\ref{thm bajaj cond 2}) and (\ref{thm bajaj cond 3})
guarantee the existence of a $(d-1)$-cycle
and an element with cycle decomposition $(2,d-3)$ in $\gal(p)$.
These two elements,
together with the transitivity of $\gal(p)$,
imply that $\gal(p)\cong S_d$.
\begin{lem}[{\cite[Chapter~4]{dummit}}] \label{thm::symmetric_is_unsolvable}
A symmetric group $S_n$ over $n$ elements is solvable if and only if $n\leq 4$.
\end{lem}

\section{Unsolvability}\label{sec::unsolvability}%

Consider the situation depicted in Fig.~\ref{figure parallel example},
\begin{figure}
\centering 
\includegraphics[scale = 0.9]{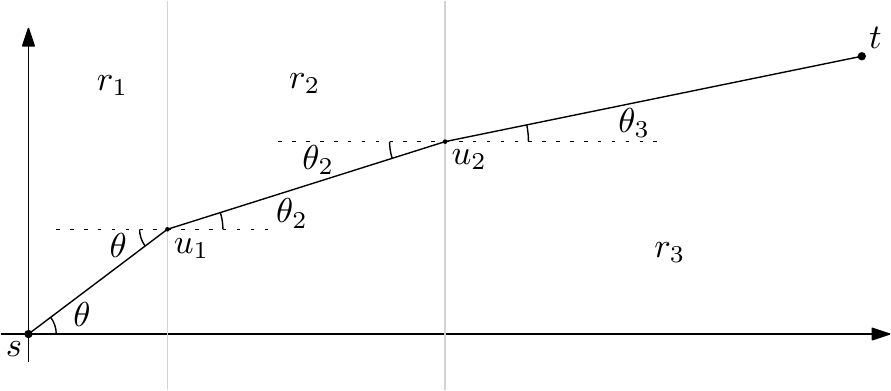}
\caption{An instance of the weighted region shortest path problem where the shortest path has two bends, namely $u_1$ and $u_2$.\label{figure parallel example}}
\end{figure}%
where $s = (0,0)$ is
the source and $t = (6,2)$ is the target. The three regions $r_1$, $r_2$
and $r_3$ have weights $w_1 = 1$, $w_2 = 2$ and $w_3 = 3$, respectively.
The three regions are $r_1=\{(x,y)\in\mathbb{R}^2\mid x\leq 1\}$,
$r_2=\{(x,y)\in\mathbb{R}^2\mid 1\leq x\leq 3\}$
and $r_3=\{(x,y)\in\mathbb{R}^2\mid x \geq 3\}$.

The optimal path satisfies Snell-Descartes law~\cite{mitchell1991weighted}.
We denote by $\theta_i$ the angle made by the incident ray in $r_i$
($1 \leq i \leq 3$).
For simplicity,
we let $\theta = \theta_1$.
Hence,
we must have $\sin(\theta_2) = \frac{w_1}{w_2}\sin(\theta)$ and $\sin(\theta_3) = \frac{w_1}{w_3}\sin(\theta)$.

Since the sum of the vertical distances travelled in all regions must
be equal to the $y$-coordinate of $t$,
we need to solve
$$\tan(\theta) + 2\tan(\theta_2) + 3\tan(\theta_3) = 2.$$
Since $\tan(\theta) = \frac{\sin(\theta)}{\sqrt{1-\sin^2\!(\theta)}}$ for $0 \leq \theta < \frac{1}{2}\pi$,
this can be rewritten as
$$\phi(X)=\frac{X}{\sqrt{1-X^2}} + 2\frac{\frac{w_1}{w_2}X}{\sqrt{1-\left(\frac{w_1}{w_2}X\right)^2}} + 3\frac{\frac{w_1}{w_3}X}{\sqrt{1-\left(\frac{w_1}{w_3}X\right)^2}} = 2 \enspace,$$
where $X = \sin(\theta)$.
By appropriately squaring three times,
this can be transformed into
\begin{align}
\label{eqn:p12}
p_{12}(u)\;=\;&419904-3545856u+12394944u^2-24006816u^3+28904608u^4-22882588u^5\\
\nonumber
&+12204109u^6-4396586u^7+1060979u^8-168272u^9+16843u^{10}-970u^{11}+25u^{12} = 0\enspace,
\end{align}
where $u = \sqrt{X}$.
\begin{thm}
The weighted region shortest path problem cannot be
solved exactly within the {\sf ACM$\bbbq$}.
\end{thm}
\begin{proof}
Following the notation of Theorem~\ref{theorem bajaj}, and the above example, we have
$p(u) = p_{12}(u)$, $d = 12$ and $\Delta(p) = 2^{57}\cdot 3^{98} \cdot 5^{22} \cdot 1847\cdot 814585609$.

With numerical methods,
one finds that
for $0 \leq \theta < \frac{1}{2}\pi$,
there exists a unique number $\alpha$ such that $\phi(\alpha) = 2$.
This number $\alpha$ is such that $0.60206 < \alpha < 0.60208$.

However,
with $q_1=79$, $q_2=31$ and $q_3=11$,
one finds
{\small
\begin{eqnarray*}
p(x) &\equiv& 19+59u+2u^2+20u^3+9u^4+78u^5+31u^6+u^7+9u^8+77u^9+16u^{10}+57u^{11}+25u^{12} \pmod{79},\\
p(x) &\equiv& 25(20+u)(10+27u+u^2+6u^3+30u^4+14u^5+5u^6+28u^7+12u^8+17u^9+28u^{10}+u^{11}) \pmod{31},\\
p(x) &\equiv& 3(u+9)(8+u+u^2)(8+10u+3u^2+6u^3+3u^4+u^5+6u^7+4u^8+u^9) \pmod{11}.
\end{eqnarray*}
}
Therefore,
$\gal(p) \cong S_{12}$ by Theorem~\ref{theorem bajaj}.
Moreover,
Lemma~\ref{thm::symmetric_is_unsolvable} tells us that $S_{12}$
is non-solvable.

Hence,
$\alpha$ cannot be computed within the {\sf ACM$\bbbq$}
otherwise this would contradict the non-solvability of $S_{12}$.
Therefore, in general, the weighted region shortest path problem
cannot be solved exactly within the {\sf ACM$\bbbq$}.
\end{proof}

\begin{obs}
If a problem is solvable within the {\sf ACM\(\bbbq\)}, then we can express its solution as a finite sequence of the allowed operations on the rational input data. For practical applications however, we may need to rely on approximations of such an explicit representation, due to the occurrence of roots. The latter can hardly be avoided for the weighted region shortest path problem, as the length of a path is the weighted sum of Euclidean distances. A problem may be unsolvable in the {\sf ACM\(\bbbq\)} even though its solution can be approximated with sufficient precision in practice. Nonetheless, we use the {\sf ACM\(\bbbq\)} as a viewpoint to gain insights about algebraic complexity and applicability of symbolic computation. One of the advantages of symbolic computation is the reusability of a result without cascaded approximation error. Unsolvability on the other hand concludes any search for a closed formula for solutions and provides further justification for the employment of approximation approaches.
\end{obs}

\begin{obs}
Let $\mathcal{P}$ be a problem that can be translated into a (system of) polynomial equation(s), and assume that we want to use 
Theorem~\ref{theorem bajaj} to prove that $\mathcal{P}$
cannot be solved exactly within the {\sf ACM$\bbbq$}. 
In general,
$\mathcal{P}$ admits infinitely many different instances leading to
infinitely many different polynomial equations.
Our experience shows that most of the time, Theorem~\ref{theorem bajaj} applies on the first instance of $\mathcal{P}$
we can think of. Otherwise, one can use a symbolic computation software as a black box and compute
$\gal(p)$. To use Theorem~\ref{theorem bajaj}, we need to find three prime numbers that
satisfy the constraining properties. Bajaj~\cite{DBLP:journals/dcg/Bajaj88}
explains why trying $d+1$ prime numbers that do not divide $\Delta(p)$ will
most likely be sufficient. As for the factorization of a polynomial modulo a
prime number, refer to~\cite{Cohen:1995:CCA:206777} for standard algorithms
that perform this task.
\end{obs}

\section{Generalization to $n$ Regions}
We have shown that one instance
of the weighted region shortest path problem is unsolvable,
which shows this problem is unsolvable in general.
One usual way of getting around this problem is to assume that
we work in a model of computation where it takes $O(1)$ time
to solve any polynomial equation of bounded degree.
However,
we can extend our example to $n$ regions,
for arbitrarily large values of $n$,
where we get to solve
$$\tan(\theta) + 2\tan(\theta_2) + ... + n\tan(\theta_n) = 2,$$
which leads to
$$\frac{X}{\sqrt{1-X^2}} + 2\frac{\frac{w_1}{w_2}X}{\sqrt{1-\left(\frac{w_1}{w_2}X\right)^2}} + ... + n\frac{\frac{w_1}{w_n}X}{\sqrt{1-\left(\frac{w_1}{w_n}X\right)^2}} = 2.$$
This last equation can be transformed into a polynomial equation of degree $n\,2^{n-1}$.
Hence,
the degree of the polynomial equations involved in this problem is unbounded.

It would be useful to know how likely is it for an instance
of the weighted region shortest path problem to be unsolvable.
If we know the sequence of regions that the shortest path goes through,
then we know that the path itself is made up of a sequence of line segments
passing through the interiors of the prescribed regions and bending only on the boundaries of these regions.  Furthermore, the shortest path is locally optimal between any two bendpoints.  That is, if we treat the bendpoints $u_i$ and $u_{i+3}$ as fixed, then the intermediate bendpoints $u_{i+1}$ and $u_{i+2}$ must be optimal with respect to $u_i$ and $u_{i+3}$.  This implies that any instance of the weighted region shortest path problem in which the shortest path goes through at least three regions, will contain a generalization of the given counter-example. In particular, the equations involved in the solution will have the same form, but with different coefficients. This will be true, except in very specific cases. Thus, a generic instance of the weighted region shortest path problem in which the path passes through at least three regions is more likely to be unsolvable.

\section{Conclusions and Future Work} \label{sec::conclusions}

The method we employed, Bajaj's technique, will be a  useful tool-kit to prove similar unsolvability results and guide  more realistic analysis of problems in computational geometry with algebraic components.
When the degree of the polynomial equations involved in the solution of a problem is unbounded,
then an unsolvability result like the one presented in this paper justifies the search for an approximate solution.

\bibliographystyle{plain}
\bibliography{ARXIV_unsolvability_wrsp}

\end{document}